\documentclass[leqno,a4paper,twoside,11pt]{article}
\usepackage[nointlimits,nosumlimits]{amsmath}
\usepackage{amsfonts,amssymb,amsthm,ifthen}
\usepackage[utf8x]{inputenc}
\usepackage[square,numbers]{natbib}

\newtheorem{Thm}{Theorem}[section]

\newtheorem{Lem}[Thm]{Lemma}

\newtheorem{Def}[Thm]{Definition}

\newtheorem{Dlm}[Thm]{Definition and Lemma}

\theoremstyle{definition}

\newcommand{\arr}[2]{%
  \begin{array}{@{}#1@{}}#2\end{array}}
\newcommand{\abs}[1]{\left| #1 \right|}
\newcommand{\scal}[3][]{\ifthenelse{\equal{#1}{}}{
  \left\langle #2,\,#3 \right\rangle
}{\ifthenelse{\equal{#1}{(}}{
  \left( #2,\,#3 \right)
}{\ifthenelse{\equal{#1}{[}}{
  \left[ #2,\,#3 \right]
}{
  #1\left( #2,\,#3 \right)
}}}}

\newcommand{\setsep}{\;\big|\;}
\newcommand{\dd}[2][]{\frac{\partial #1}{\partial #2}}

\renewcommand{\div}{\mathrm{div}}
\newcommand{\tr}{\mathrm{tr}}
\newcommand{\id}{\mathrm{id}}
\newcommand{\ev}{\mathrm{ev}}
\newcommand{\Der}{\mathrm{Der}}
\newcommand{\Hom}{\mathrm{Hom}}
\newcommand{\str}{\mathrm{str}}
\newcommand{\dvol}{\mathrm{dvol}}
\newcommand{\sdet}{\mathrm{sdet}}
\newcommand{\dsvol}{\mathrm{dsvol}}

\newcommand{\mA}{\mathcal A}
\newcommand{\mD}{\mathcal D}
\newcommand{\mF}{\mathcal F}
\newcommand{\mO}{\mathcal O}
\newcommand{\mS}{\mathcal S}
\newcommand{\bN}{\mathbb N}
\newcommand{\bR}{\mathbb R}
\newcommand{\bZ}{\mathbb Z}
\newcommand{\fgl}{\mathfrak{gl}}

\renewcommand{\title}[1]{\vbox{\center\LARGE{\textsc{#1}}}\vspace{5mm}}
\renewcommand{\author}[1]{\vbox{\center\large{\textsc{#1}}}\vspace{5mm}}
\newcommand{\address}[1]{\vbox{\center\em#1}}
\newcommand{\email}[1]{\vbox{\center\tt#1}\vspace{5mm}}

\begin{document}

\title{Killing Vector Fields and\\ Superharmonic Field Theories}

\author{Josua Groeger$^1$}

\address{Humboldt-Universit\"at zu Berlin, Institut f\"ur Mathematik,\\
  Rudower Chaussee 25, 12489 Berlin, Germany }

\email{$^1$groegerj@mathematik.hu-berlin.de}

\begin{abstract}
\noindent
The harmonic action functional allows a natural generalisation
to semi-Riemannian supergeometry, referred to as superharmonic action,
which resembles the supersymmetric sigma models studied in high energy physics.
We show that Killing vector fields are infinitesimal supersymmetries
of the superharmonic action and prove three different Noether theorems
in this context. En passant, we provide a homogeneous treatment of
five characterisations of Killing vector fields on
semi-Riemannian supermanifolds, thus filling a gap in the literature.
\end{abstract}

\section{Introduction}

Symmetries belong to the main ingredients of modern physical theories.
In high energy physics, supersymmetry is a conjectured transition between
the two types of elementary particles in nature: bosons and fermions, which differ in their statistics.
The observed properties of these particles are, to a good agreement, described by quantum field theories,
see e.g. \cite{PS95} for a standard treatment. The quantum theories are, usually, based on classical field theories. In this context, bosons and fermions
are described by even and odd fields, respectively.
Free bosonic string theory is modelled on the well-known harmonic action \cite{Jos01}, while a class of supersymmetric
extensions is given by so called supersymmetric sigma models \cite{DF99b}.

In this article, we study a similar such extension which, from a mathematical
point of view, is more natural and also more general in that it can be formulated
for every pair of semi-Riemannian supermanifolds. We shall, therefore, refer to the resulting
field theories as superharmonic. In contrast to the aforementioned sigma models, the superharmonic
action always allows a ''superspace formulation'', meaning that all fields occuring can be included
in a single superfield,
a mathematical model of which is a map with flesh \cite{Hel09}.

Killing vector fields are known to be infinitesimal symmetries of the harmonic action.
According to the Noether principle, every symmetry of a classical field theory should induce a
conserved quantity. In the harmonic theory, this is indeed the case (\cite{BE81}, \cite{Hel02}),
while it is a priori not clear for the superharmonic action. Indeed,
whereas certain elements of supergeometry, such as the divergence of a vector field, differ from the classical theory
of manifolds merely by a number of signs, others are non-trivial extensions. Examples for the latter include integration theory
and maps with flesh.

In the present article, we show that Killing vector fields on the domain as well as on the target
space supermanifold are infinitesimal symmetries of the superharmonic action. As our main results,
we formulate and prove three different Noether theorems in this context. This is the subject matter
of Sec. \ref{secSuperharmonic}.
There are several equivalent definitions for a Killing vector field on a semi-Riemannian manifold,
most of which have been generalised to supergeometry (\cite{MSV96}, \cite{Goe08}, \cite{ACDS97}),
while a homogeneous treatment of the subject is yet missing. Another aim of this article is
to fill this gap. In Sec. \ref{secKillingVF}, we prove equivalence of five characterisations of
Killing vector fields. It turns out that the definition used in \cite{ACDS97} is a non-equivalent variation.
To start with, Sec. \ref{secSemiRiemSupermf} reviews elements of the theory of semi-Riemannian supermanifolds
needed later.

\section{Semi-Riemannian Supermanifolds}
\label{secSemiRiemSupermf}

Superharmonic field theories are based on semi-Riemannian supermanifolds. For later
use, we briefly recall the relevant background here.

Throughout the article,
we adopt the Berezin-Kostant-Leites definition of supermanifolds and their morphisms
in terms of sheaves as in \cite{Lei80}. In particular, a supermanifold is a ringed space
$(M,\mO_M)$, and a morphism $\Phi:(M,\mO_M)\rightarrow(N,\mO_N)$ of supermanifolds
consists of two parts $\Phi=(\varphi,\phi)$. We shall occasionally abuse notation
and write $M$ instead of $(M,\mO_M)$. Modern monographs on the general theory of supermanifolds
include \cite{Var04} and \cite{CCF11} while aspects of Riemannian supergeometry are studied in \cite{Goe08}.
References for more specialised topics will be given at suitable positions throughout the text.

Following the conventions of \cite{Gro11a}, we denote the (super) tangent sheaf,
i.e. the sheaf of superderivations of $\mO_M$, by $\mS M:=\mathrm{Der}(\mO_M)$.
The (super) tangent space at a point $p\in M$ is defined by
\begin{align*}
S_pM:=\{v:(\mO_M)_p\rightarrow\bR\setsep\bR\mathrm{-linear}\,,\;v(fg)=v(f)\tilde{g}(p)+(-1)^{\abs{v}\abs{f}}\tilde{f}(p)v(g)\}
\end{align*}
where $(\mO_M)_p$ is the stalk of $\mO_M$ at $p$ and tilde denotes the canonical projection
$\mO_M\rightarrow C^{\infty}_M\,,\;f\mapsto\tilde{f}$ by evaluation $\tilde{f}(p)=\ev|_pf$.
Any vector field $X\in\mS M$ gives rise to the tangent space valued map $p\mapsto X(p)\in S_pM$ via
\begin{align*}
X(p)(f):=\widetilde{X(f)}(p)\;,\qquad f\in(\mO_M)_p
\end{align*}
Note the use of the shorthand notation $X\in\mS M$, meaning $X\in\mS M(U)$ for $p\in U\subseteq M$.
Similarly, we yield a canonical (classical) vector field $\tilde{X}\in\Gamma(TM)$ on $M$ by setting
$\tilde{X}(\tilde{f}):=\widetilde{X_0(f)}$. It is well-known that superfunctions $f$ are not detemined
by their values $\tilde{f}(p)$ and, likewise, vector fields $X$ are not determined by their
values $X(p)$.
With respect to local coordinates $(\xi^1,\ldots,\xi^n,\xi^{n+1},\ldots,\xi^{n+m})$ of $M$, the tuple
$(\dd{\xi^1},\ldots,\dd{\xi^{n+m}})$ is a local $\mO_M$-basis of $\mS M$ which is adapted in the sense
that the first $n$ vector fields are even and the remaining $m$ ones are odd.
Likewise, the tuple $(\dd{\xi^1}(p),\ldots,\dd{\xi^{n+m}}(p))$ is an adapted basis of the super vector space $S_pM$.

Tensor calculus on supermanifolds is based on superlinear algebra. For $U\subseteq M$ sufficiently small,
$\mS M(U)$ is a free $\mO_M(U)$-supermodule of rank $n|m=\dim M$.
In general, consider a supercommutative superalgebra $A$ and free $A$-supermodules $M$ and $N$ of
rank $m|n$ and $r|s$, respectively.
With respect to adapted right bases $(f_1,\ldots,f_{n+m})$ and $(g_1,\ldots,g_{r+s})$ of $M$ and $N$, respectively,
any superlinear map $L:M\rightarrow N$ can be identified with a matrix $L\in\mathrm{Mat}_A(m|n,r|s)$.
Similarly we denote, for an even (super-)bilinear form $B\in\Hom_A(M\otimes_A M,A)$, the corresponding matrix
with entries $B_{jk}:=B(f_j,f_k)$ by the same symbol $B\in\mathrm{Mat}(n|m,A)$.

Let $GL_{n|m}(A)$ denote the group of even and invertible $n|m$-matrices with entries in $A$.
The orthosymplectic group of dimension $(t,s)|2m$ is defined as follows.
\begin{align*}
OSp_{(t,s)|2m}(A):=\{L\in GL_{t+s|2m}(A)\setsep\forall v,w\in A^{t+s|2m}:\,g_0(Lv,Lw)=g_0(v,w)\}
\end{align*}
where $g_0$ denotes the standard supermetric which, with respect to the standard basis for $A^{t+s|2m}$,
is given by the matrix
\begin{align}
\label{eqnStandardMetric}
\renewcommand{\arraystretch}{1.5}
g_0&:=\left(\arr{@{\;}c@{\;}|@{\;}c@{\;}}
{G_{t,s}&0\\\hline 0&J_{2m}}\right)\quad\mathrm{where}\\
G_{t,s}&:=\left(\arr{cccc}{-1_{t\times t}&0\\0&1_{s\times s}}\right)\;,\quad
J_{2m}:=\left(\arr{ccc}{J_2&0&0\\0&\ddots&0\\0&0&J_2}\right)\;,\quad
J_2:=\left(\arr{cc}{0&-1\\1&0}\right)\nonumber
\end{align}
The corresponding super Lie algebra is
\begin{align*}
\mathfrak{osp}_{(t,s)|2m}(A):=\{L\in\fgl_{t+s|2m}(A)\setsep \scal[g_0]{Lv}{w}=-(-1)^{\abs{L}\abs{v}}\scal[g_0]{v}{Lw}\}
\end{align*}
where $\fgl_{n|m}(A):=\fgl_{n|m}\otimes A$ is the super Lie algebra of all matrices with entries in $A$.
For $A=\bR$, we define $\mathfrak{osp}_{(t,s)|2m}:=\mathfrak{osp}_{(t,s)|2m}(\bR)$.
By means of choosing a basis $\{T^j\}$ of $\mathfrak{osp}_{(t,s)|2m}$, it is easy to
see that $\mathfrak{osp}_{(t,s)|2m}(A)=\mathfrak{osp}_{(t,s)|2m}\otimes A$.

Now let $g$ be an even (i.e. parity-preserving), nondegenerate and supersymmetric bilinear form (a supermetric for short).
By an extension of the Gram-Schmidt procedure as detailed e.g. in Sec. 2.8 of \cite{DeW84},
there is an adapted basis $(e_1,\ldots,e_{t+s+2m})$ of $M$,
such that $g=g_0$ on the level of matrices, which we shall call an $OSp_{(t,s)|2m}$-basis.
In particular, the values of $t,s,m\in\bN$ are uniquely determined by $g$.
As in Sec. 3.5 of \cite{Han12} with slightly different conventions, we introduce the even map $J$, which is
defined with respect to an $OSp_{(t,s)|2m}$-basis $\{e_i\}$ as follows.
\begin{align*}
Je_k:=\left\{\arr{ll}{-e_k&k\leq t\\e_k&t<k\leq t+s\\e_{k+1}&k=t+s+2l-1\\-e_{k-1}&k=t+s+2l}\right.
\end{align*}
This is such that $\scal[g]{e_k}{Je_j}=(-1)^{\abs{e_k}}\delta_{kj}$ and $Je_k=(-1)^{\abs{e_k}}h_{km}e_m$
and, moreover, every $v\in M$ has the expansion
\begin{align*}
v=\scal[g]{v}{e_j}Je_j=(-1)^{\abs{e_j}}\scal[g]{v}{Je_j}e_j
\end{align*}

$g$ identifies any other bilinear form $K\in\Hom_A(M\otimes_A M,A)$ with a superlinear map $\tilde{K}:M\rightarrow M$
of the same parity as $K$ via
\begin{align*}
\scal[g]{\tilde{K}(v)}{w}=\scal[K]{v}{w}
\end{align*}
The supertrace of $K$ with respect to $g$ is defined as
\begin{align*}
\str_gK:=\str\tilde{K}=(-1)^{\abs{f_j}(\abs{\tilde{K}}+1)}\tilde{K}^j_{\phantom{j}j}
\end{align*}
where the second equation holds for any adapted (right) basis $\{f_j\}$ of $M$ upon identifying
$\tilde{K}$ with a supermatrix.
An explicit calculation then shows that, regardless of the parity of $K$,
the following formula holds for every $OSp_{(t,s)|2m}$-basis $f_j=e_j$.
\begin{align}
\label{eqnStraceMetric}
\str_gK=\scal[K]{e_j}{Je_j}=(-1)^{\abs{e_j}}\scal[K]{Je_j}{e_j}
\end{align}

We shall use the preceding superlinear algebra for the tensor calculus on a supermanifold $(M,\mO_M)$.
Let $\mathrm{End}_{\mO_M}(\mS M)$ denote the sheaf of superlinear endomorphisms.
As usual, we write $E\in\mathrm{End}_{\mO_M}(\mS M)$ for a sheaf morphism $E=\{E_U\}_{\stackrel{U\subseteq M}{\mathrm{open}}}$
which, by a slight abuse of notation, we shall call a \emph{section} of the sheaf of endomorphisms.
Let, similarly, $\Hom_{\mO_M}(\mS M,\mO_M)$ and $\Hom_{\mO_M}(\mS M\otimes_{\mO_M}\mS M,\mO_M)$ denote
the sheaves of superlinear maps and superbilinear maps, respectively.
We identify the tensor product of two one-forms $F,G\in\Hom_{\mO_M}(\mS M,\mO_M)$ with a bilinear form via
\begin{align*}
(F\otimes G)(X,Y):=(-1)^{\abs{G}\abs{X}}F(X)\cdot G(Y)
\end{align*}
which can be taken as a definition \cite{DM99}.
Sections of any sheaf of multilinear forms will be commonly denoted as \emph{tensors}
or \emph{tensor fields}.
The differential $df$ of a superfunction $f\in\mO_M$ is defined by the formula $df[X]:=(-1)^{\abs{f}\abs{X}}X(f)$
for $X\in\mS M$. In particular, we may consider the differential of coordinate functions $\xi^i$.
Our sign conventions are such that any even bilinear form $B\in\Hom_{\mO_M}(\mS M\otimes_{\mO_M}\mS M,\mO_M)$
has the local form
\begin{align}
\label{eqnLocalBilinear}
B=(-1)^{\abs{\xi^i}+\abs{\xi^j}+\abs{\xi^i}\abs{\xi^j}}B_{ij}\cdot d\xi^i\otimes d\xi^j
\end{align}
An even bilinear form $g$ which is non-degenerate and (super-)symmetric is called a \emph{semi-Riemannian supermetric}.
Locally, the above treatment on superlinear algebra applies here. In particular, there exists an $OSp_{(t,s)|2m}$-basis
$\{e_j\}$ with $t$, $s$ and $m$ being intrinsic invariants of $g$.

Now consider a morphism $\Phi=(\varphi,\phi):(M,\mO_M)\rightarrow(N,\mO_N)$ of supermanifolds.
Its \emph{differential} is the morphism of sheaves
\begin{align}
\label{eqnDifferential}
d\Phi:\varphi_*\mS M\rightarrow\mS\Phi\;,\qquad d\Phi(Y):=Y\circ\phi
\end{align}
where $\mS\Phi:=\Der(\mO_N,\varphi_*\mO_M)$ denotes the sheaf of derivations (vector fields) along $\Phi$,
which is locally free of rank the dimension of $N$.
The pullback of tensors on $N$ is defined as follows. Following the conventions
used in \cite{Gro11a}, let $E\in\mathrm{End}_{\mO_N}(\mS N)$ and $F\in\Hom_{\mO_N}(\mS N,\mO_N)$
and $B\in\Hom_{\mO_N}(\mS N\otimes_{\mO_N}\mS N,\mO_N)$ be tensor fields.
Now, prescribing
\begin{align*}
E_{\Phi}(\phi\circ Y):=\phi\circ E(Y)\;,\quad
F_{\Phi}(\phi\circ Y):=\phi\circ F(Y)\;,\quad
\scal[B_{\Phi}]{\phi\circ Y}{\phi\circ Z}:=\phi\circ\scal[B]{Y}{Z}
\end{align*}
for $Y,Z\in\mS N$, together with super(bi)linear extensions for general sections of $\mS\Phi$, yields
well-defined sections $E_{\Phi}\in\mathrm{End}_{\varphi_*\mO_M}(\mS\Phi)$ and
$F_{\Phi}\in\mathrm{Hom}_{\varphi_*\mO_M}(\mS\Phi,\varphi_*\mO_M)$ and
$B_{\Phi}\in\Hom_{\varphi_*\mO_M}(\mS\Phi\otimes_{\varphi_*\mO_M}\mS\Phi,\varphi_*\mO_M)$, respectively.

A \emph{connection} on $N$ is an even $\bR$-linear sheaf morphism $\nabla:\mS N\rightarrow\mS^*N\otimes_{\mO_N}\mS N$
satisfying the (graded) Leibniz rule. If $(N,g)$ is a semi-Riemannian supermanifold, there is a unique
connection which is (graded) metric and torsion-free, called the Levi-Civita connection \cite{Goe08}.
In general, for a connection $\nabla$ on $N$, there is a canonical pullback connection
$\nabla_{\Phi}:\mS\Phi\rightarrow\varphi_*\mS^*M\otimes_{\varphi_*\mO_M}\mS\Phi$ (see \cite{GW12}),
that, in the following, we shall denote simply by $\nabla$.
With respect to local coordinates $\{\eta^i\}$ on $N$, it reads
\begin{align}
\label{eqnPullbackConnection}
\nabla_X\left((\phi\circ\partial_{\eta^j})\cdot Y^j\right)
=(-1)^{\abs{X}\abs{\eta^j}}(\phi\circ\partial_{\eta^j})\cdot X(Y^j)
+X(\phi\circ\eta^i)\cdot(\phi\circ\nabla_{\partial_{\eta^i}}\partial_{\eta^j})\cdot Y^j
\end{align}

\begin{Lem}
\label{lemMetricPullback}
Let $(N,g)$ be a semi-Riemannian supermanifold and $\nabla$ a superconnection on $N$ which is metric.
Then the pullback $\nabla_{\Phi}$ is metric in the following sense.
\begin{align*}
X\scal[g_{\Phi}]{Y}{Z}=\scal[g_{\Phi}]{(\nabla_{\Phi})_X Y}{Z}+(-1)^{\abs{X}\abs{Y}}\scal[g_{\Phi}]{Y}{(\nabla_{\Phi})_X Z}
\end{align*}
holds true for every $X\in\varphi_*\mS M$ as well as $Y,Z\in\mS\Phi$.
\end{Lem}

\begin{proof}
This follows from a straightforward calculation in local coordinates using (\ref{eqnPullbackConnection}).
\end{proof}

There is a second type of pullback for tensors and functions $f\in\mO_N$ on $N$, which yields
corresponding objects on $M$. We define
\begin{align}
\Phi^*f:=\phi(f)\;,\qquad
\Phi^*F(X):=F_{\Phi}(d\Phi[X])\;,\qquad
\scal[\Phi^*B]{X}{Y}:=\scal[B_{\Phi}]{d\Phi[X]}{d\Phi[Y]}
\end{align}
such that
\begin{align*}
\Phi^*B\in\Hom_{\varphi_*\mO_M}(\varphi_*\mS M\otimes\varphi_*\mS M,\varphi_*\mO_M)
\cong\Hom_{\mO_M}(\mS M\otimes\mS M,\mO_M)
\end{align*}
and analogous for other tensors. The canonical identification with an (ordinary) tensor on $M$ stated
follows from the general theory of ringed spaces, e.g. from Thm. 4.4.14 of \cite{Ten75} applied
to the induced morphism $(\varphi,\id):(M,\mO_M)\rightarrow(N,\varphi_*\mO_M)$ of ringed spaces.
A direct calculation yields
\begin{align}
\label{eqnPullbackMultiplicationTensor}
\Phi^*(f\cdot F)=(\Phi^*f)\cdot(\Phi^*F)\;,\qquad
\Phi^*(F\otimes G)=\Phi^*F\otimes\Phi^*G
\end{align}
Moreover, one verifies that, in case $\Phi$ is a diffeomorphism, it holds
\begin{align*}
\scal[\Phi^*B]{X}{Y}&=\phi\circ\scal[B]{(\phi^{-1})\circ X\circ\phi}{(\phi^{-1})\circ Y\circ\phi}
\end{align*}
which is the definition of the pullback in \cite{Goe08}. The next lemma will be needed in calculations
below.

\begin{Lem}
\label{lemPullbackFunctions}
Let $f\in\mO_N$ be a superfunction and $Y\in\mS M$ be a super vector field. Then
\begin{align*}
\Phi^*df[Y]=(-1)^{\abs{f}\abs{Y}}d\Phi[Y](f)
\end{align*}
\end{Lem}

\begin{proof}
The assertion holds by the following calculation in coordinates $\{\eta^i\}$ on $N$.
\begin{align*}
\Phi^*df[Y]&=df_{\Phi}(d\Phi[Y])\\
&=df_{\Phi}\left((\phi\circ\partial_{\eta^i})d\Phi[Y]^i\right)\\
&=\phi\circ df[\partial_{\eta_i}]\cdot d\Phi[Y]^i\\
&=(-1)^{\abs{f}\abs{\eta_i}}(\phi\circ\partial_{\eta_i})(f)\cdot d\Phi[Y]^i\\
&=(-1)^{\abs{f}\abs{\eta_i}}(-1)^{\abs{f}(d\Phi[Y]^i)}\left((\phi\circ\partial_{\eta_i})\cdot d\Phi[Y]^i\right)(f)\\
&=(-1)^{\abs{f}\abs{Y}}d\Phi[Y](f)
\end{align*}
\end{proof}

The matrix groups $GL_{n|m}(A)$ and $OSp_{(t,s)|2m}(A)$ give rise to super Lie groups
$GL_{n|m}$ and $OSp_{(t,s)|2m}$, respectively. They are examples of matrix super Lie groups
and, for the treatment of $G$-structures below, are most conveniently described in terms
of $S$-points.

In general, to a supermanifold $M$ we
can associate its functor of points $M(\cdot):\mathrm{SMan}^{\mathrm{Op}}\rightarrow\mathrm{Set}$
by sending any supermanifold $S$ to its \emph{$S$-point} $M(S):=\Hom(S,M)$ and any morphism $\Phi:T\rightarrow S$
to $M(\Phi):M(S)\rightarrow M(T)$ through $m\mapsto m\circ\Phi$.
By Yoneda's lemma, morphisms $\psi:M\rightarrow N$ are in bijection with
natural transformations $\psi(\cdot):M(\cdot)\rightarrow N(\cdot)$.
This induces a canonical embedding of the category $\mathrm{SMan}$ into the functor category
$[\mathrm{SMan}^{\mathrm{Op}},\mathrm{Set}]$, elements of which are called \emph{generalised supermanifolds},
and those in the image of the embedding are called \emph{representable}.
A \emph{generalised super Lie group} is an object of the functor category
$[\mathrm{SMan}^{\mathrm{Op}},\mathrm{Grp}]$ (which canonically embeds into $[\mathrm{SMan}^{\mathrm{Op}},\mathrm{Set}]$).
The representing supermanifold, if existing, is called a \emph{super Lie group} which,
equivalently, can also be defined as a group object in the category of supermanifolds.

$GL_{n|m}$ is the functor which sends any supermanifold $S$ to the
multiplicative group $GL_{n|m}(\mO_S(S))$ of even and invertible $n|m$-matrices with values in $\mO_S(S)$
and any morphism $T\rightarrow S$ to the induced map $GL_{n|m}(\mO_S(S))\rightarrow GL_{n|m}(\mO_T(T))$.
This functor is representable such that, on the level of $S$-points, multiplication is ordinary matrix multiplication.
The same construction can be applied to natural subgroups of $GL_{n|m}(A)$.
For example, the generalised super Lie group $OSp_{(t,s)|2m}$ is the functor defined through
$S\mapsto OSp_{(t,s)|2m}(\mO_S(S))$ and the corresponding morphism map.
It is representable by the stabliliser condition \cite{BCF09}.
For our applications, it will be sufficient to consider all super Lie groups occuring as
generalised. This is also the point of view of \cite{ACDS97}.

We shall next describe frame fields and $G$-structures on a supermanifold, where $G$ is
a super Lie group. In particular, we will see that supermetrics are equivalent to $OSp_{(t,s)|2m}$-structures.

A tuple $(X_1,\ldots,X_n,X_{n+1},\ldots,X_{n+m})$ of $n$ even and $m$ odd vector fields
on an open set $U\subseteq M$ is called a \emph{frame field}, if the tuple $(X_1(p),\ldots,X_{n+m}(p))$
of tangent vectors is a basis of the supervector space $T_pM$ for every $p\in U$.
If $U$ is sufficiently small (such that it lies in a coordinate chart), this condition
is equivalent to $(X_1,\ldots,X_{n+m})$ being an adapted $\mO_M(U)$-basis.
The frame fields form a sheaf on $M$, denoted $U\mapsto\mF(U)$.

The super Lie group $GL_{n|m}$ induces a sheaf of groups over $M$ via $GL_{n|m}(U):=GL_{n|m}(\mO_M(U))$,
where the latter corresponds to a morphism $U\rightarrow GL_{n|m}$ of supermanifolds.
$GL_{n|m}(U)$ acts naturally from the right on the frame fields $\mF(U)$. Explicitly,
\begin{align}
\label{eqnRightAction}
(X_1,\ldots,X_{n+m})\cdot A:=\left(\sum_i X_i\cdot A_{i1},\ldots,\sum_i X_i\cdot A_{in}\right)
\end{align}
where $A_{jk}$ denotes the $jk$-th entry of the matrix of $A\in GL_{n|m}(U)$.
If $U$ is sufficiently small, we see that this action is simply transitive.

A supermetric $g$ defines the subsheaf of $OSp_{(t,s)|2m}$-frames as follows.
Denoting the standard basis of $A^{t+s|2m}$ by $\{e_i\}$, we set
\begin{align*}
\mF_g(U):=\{(X_1,\ldots)\in\mF(U)\setsep g(X_i,X_j)=g_0(e_i,e_j)\}
\end{align*}
The orthosymplectic supergroup $OSp_{(t,s)|2m}$ acts on $\mF_g$ via (\ref{eqnRightAction}).
We already know that, given two frames $(X_1,\ldots),(Y_1,\ldots)\in\mF_g(U)$ for $U$ sufficiently small,
there is a unique $A\in GL_{t+s|2m}(U)$ such that $(X_1,\ldots)\cdot A=(Y_1,\ldots)$.
But this is exactly the condition for $A\in OSp_{(t,s)|2m}$.

\begin{Def}[\cite{ACDS97}]
Let $G\subseteq GL_{n|m}$ be a Lie subgroup. A \emph{$G$-structure} on a supermanifold $M$ is a sheaf
$\mF_G$ of subsets $\mF_G\subseteq\mF$ such that $G(U)$ acts on $\mF_G(U)$ and for all points,
there is a neighbourhood for which the action is simply transitive.
\end{Def}

Thus, in particular, every supermetric $g$ defines an $Osp_{(t,s)|2m}$-structure $\mF_g$. Conversely, assume that
$\mF_{OSp}$ is an $Osp_{(t,s)|2m}$-structure. We construct a supermetric $g$ as follows. Let $(X_1,\ldots)\in\mF_{OSp}$ and
define $g$ by $g(X_i,X_j):=g_0(e_i,e_j)$ and superbilinear extension. This definition does not depend on the
chosen frame field in $\mF_{OSp}$ and is such that $\mF_g=\mF_{OSp}$.

\begin{Lem}
\label{lemSupermetricOSpStructure}
There is a bijection between supermetrics and $OSp_{(t,s)|2m}$-structures.
\end{Lem}

An automorphism of $M$ is an isomorphism $\Phi:M\rightarrow M$ of supermanifolds.
Its differential (\ref{eqnDifferential}) can be identified with a sheaf morphism
$d\Phi:\mS M\rightarrow\varphi_*^{-1}\mS M$ via $d\Phi(Y):=\phi^{-1}\circ Y\circ\phi$.
It induces a sheaf (iso)morphism $d\Phi:\mF\rightarrow\varphi_*^{-1}\mF$, denoted by the same symbol.

\begin{Def}
\renewcommand{\labelenumi}{(\roman{enumi})}
Let $\Phi:M\rightarrow M$ be an automorphism of a supermanifold $M$. Let $g$ be a semi-Riemannian supermetric
and $\mF_G$ be a $G$-structure on $M$.
\begin{enumerate}
\item $\Phi$ is called an automorphism of $g$ if it is an isometry $\Phi^*g=g$.
\item $\Phi$ is called an automorphism of $\mF_G$ if $d\Phi\mF_G\subseteq\varphi_*^{-1}\mF_G$.
\end{enumerate}
\end{Def}

\begin{Lem}
\label{lemSupermetricAutomorphism}
An automorphism $\Phi:M\rightarrow M$ of $M$ is an automorphism of $g$ if and only if it is an automorphism of $\mF_g$.
\end{Lem}

\begin{proof}
Let $(X_1,\ldots)\in\mF_g(U)$. By definition, $\scal[g]{X_i}{X_j}=\scal[(\Phi^*g)]{X_i}{X_j}$ is equivalent
to $(d\Phi[X_1],\ldots)\in\varphi_*^{-1}\mF_g$, thus proving one direction. The converse follows from the
same characterisation together with superbilinearity.
\end{proof}

\section{Killing Vector Fields}
\label{secKillingVF}

In this section, we shall give five equivalent characterisations of Killing vector fields on supermanifolds.
While the building blocks of the theory are mostly available in the literature,
a homogeneous treatment such as below is still missing, and there is sometimes some confusion about the concept.
In particular, we will see how the notion of a Killing vector field as treated in \cite{ACDS97} is different
from the canonical one. We start with the Lie derivative of tensors.

\subsection{The Lie Derivative}

The Lie derivative of a tensor on a manifold $M$ with respect to a vector field $X\in\Gamma(TM)$
is defined by means of the flow $\varphi$ of $X$ which has the defining properties
$\dd{t}\varphi_x(t)=X\circ\varphi_x(t)$ and $\varphi_x(0)=x$.
By a standard theorem \cite{War83}, every vector field $X$ possesses a unique smooth flow $\varphi:D(X)\rightarrow M$,
which is defined on an open neighbourhood $D(X)\subseteq\bR\times M$ of $\{0\}\times M$.
$X$ is called complete if $D(X)=\bR\times M$.

Now let $X$ be a vector field on a supermanifold and $\tilde{X}$ be its canonical projection. We let
$\mD(X)$ be the open subsupermanifold of $\bR^{1|1}\times (M,\mO_M)$ whose underlying smooth manifold is $D(\tilde{X})$.
If $\tilde{X}$ is complete, then $\mD(X)=\bR^{1|1}\times (M,\mO_M)$.

We recall results concerning the flow of super vector fields from \cite{MSV93}.
To integrate also odd vector fields, the derivation $\dd{t}$ needs to be endowed by an odd part
to a (super-)derivation $D$.
As argued in \cite{MSV93}, one is naturally let to one of three \emph{integration models}
$D^{(1)}=\partial_t+\partial_{\tau}$, $D^{(2)}=\partial_t+\partial_{\tau}+\tau\partial_t$
and $D^{(3)}=\partial_t+\tau\partial_{\tau}+\partial_{\tau}$, corresponding to the three different
super Lie algebra structures on $\bR^{1|1}$. Due to the $\ev$-morphism, all of the following
does, however, not depend on the choice of the integration model and, for brevity,
we shall denote either derivation by $D\in\mS(\bR^{1|1})$.
To simplify notation, we shall also denote its lift to $\mD(X)$ by the same symbol $D$.

\begin{Def}
\label{defFlow}
Let $X$ be a super vector field. Its flow is a morphism $\Phi=(\varphi,\phi):\mD(X)\rightarrow M$
such that the following equations hold
\begin{align*}
\ev|_{t=t_0}\circ D\circ\phi=\ev|_{t=t_0}\circ\phi\circ X\;,\qquad
\ev|_{t=0}\circ\phi=\id
\end{align*}
which are referred to as the flow condition and initial condition, respectively.
\end{Def}

In calculations, only the homogeneous part of $D$ with the same parity as $X$ occurs since
the flow equation splits into two equations (according to the $\bZ_2$-decomposition of $X$ and $D$),
of which one vanishes if $X$ is homogeneous. By a slight abuse of notation, we may thus write
$\abs{D}=\abs{X}$.

\begin{Thm}[\cite{MSV93}]
Every super vector field possesses a unique flow.
\end{Thm}

For the strong flow condition $D\circ\phi=\phi\circ X$ without the $\ev$ morphism to hold,
$X$ must satisfy certain conditions as shown in \cite{MSV93}.

\begin{Def}
Let $f\in\mO_M$ be a superfunction. Its Lie derivative along $X$ is
\begin{align*}
L_Xf:=\ev|_{t=0}D\circ\Phi^*f
\end{align*}
where $\Phi:\mD(X)\rightarrow M$ is the flow of $X$.
\end{Def}

\begin{Def}
Let $B\in\Hom_{\mO_M}(\mS M\otimes\ldots\otimes\mS M,\mO_M)$ be a multilinear form.
We define its Lie derivative along $X$ as
\begin{align*}
L_XB:=\ev|_{t=0}D\circ\Phi^*B\in\Hom_{\mO_M}(\mS M\otimes\ldots\otimes\mS M,\mO_M)
\end{align*}
where, on the right hand side, the pullback $\Phi^*B$ is implicitly understood to be restricted to vector fields
in $\mS M\subseteq\mS(\mD(X))$.
\end{Def}

\begin{Lem}
\label{lemLieDerivativeFunctions}
Let $f\in\mO_M$ be a superfunction and $Y\in\mS M$ be a super vector field. Then,
for $X$, $Y$ and $f$ of homogeneous parity,
\begin{align*}
L_Xf=X(f)\;,\qquad
L_Xdf[Y]=(-1)^{\abs{Y}(\abs{f}+\abs{X})}Y\circ X(f)
\end{align*}
\end{Lem}

\begin{proof}
Using the defining properties of the flow, we calculate
\begin{align*}
L_Xf=\ev|_{t=0}D\circ\phi(f)=\ev|_{t=0}\phi\circ X(f)=X(f)
\end{align*}
The second assertion holds by the following calculation.
\begin{align*}
L_Xdf[Y]&=\ev|_{t=0}D\circ\Phi^*df[Y]\\
&=(-1)^{\abs{f}\abs{Y}}\ev|_{t=0}D\circ d\Phi[Y](f)\\
&=(-1)^{\abs{f}\abs{Y}}\ev|_{t=0}D\circ Y\circ\phi(f)\\
&=(-1)^{\abs{f}\abs{Y}}(-1)^{\abs{Y}\abs{D}}Y\circ\ev|_{t=0}D\circ\phi(f)\\
&=(-1)^{\abs{f}\abs{Y}}(-1)^{\abs{Y}\abs{X}}Y\circ\ev|_{t=0}\phi\circ X(f)\\
&=(-1)^{\abs{Y}(\abs{f}+\abs{X})}Y\circ X(f)
\end{align*}
Here, we used Lem. \ref{lemPullbackFunctions} and $t,\tau$-independence of $Y$.
\end{proof}

\begin{Lem}
\label{lemInitialCondition}
Let $\Phi:\mD(X)\rightarrow M$ be the flow of a vector field $X$. Then
the initial condition implies the generalisation
\begin{align*}
\ev|_{t=0}\Phi^*B=B
\end{align*}
for multilinear forms $B$.
\end{Lem}

\begin{proof}
We prove the statement for a one-form $F$, the general case is analogous. Let $Y\in\mS M$ be a vector field.
We let $\eta^1,\ldots,\eta^{n+m}$ denote local coordinates on $M$ and $\xi^1,\ldots,\xi^{n+m+2}$ the corresponding
coordinates on $\mD(X)$ such that $\xi^i=\eta^i$ for $i\in\{1,\ldots,n+m\}$ and $\xi^{n+m+1}=t$ and $\xi^{n+m+2}=\tau$.
Then
\begin{align*}
\ev|_{t=0}\Phi^*F[Y]&=\ev|_{t=0}F_{\Phi}(d\Phi[Y])\\
&=\ev|_{t=0}F_{\Phi}\left((\phi\circ\partial_{\eta^i})\cdot d\Phi[Y]^i\right)\\
&=(-1)^{(\abs{\xi^j}+\abs{\eta^i})\abs{\eta^i}}\ev|_{t=0}\left(\phi\circ F(\partial_{\eta^i})\cdot\dd{\xi^j}\phi(\eta^i)\cdot Y^j\right)\\
&=(-1)^{(\abs{\xi^j}+\abs{\eta^i})\abs{\eta^i}}\ev|_{t=0}\left(\phi\circ F(\partial_{\eta^i})\right)\cdot
\ev|_{t=0}\left(\dd{\xi^j}\phi(\eta^i)\cdot Y^j\right)\\
&=(-1)^{(\abs{\xi^j}+\abs{\eta^i})\abs{\eta^i}}F(\partial_{\eta^i})\cdot\ev|_{t=0}\left(\dd{\xi^j}\phi(\eta^i)\cdot Y^j\right)
\end{align*}
By assumption, $Y^t=Y^{\tau}=0$ such that we may replace $\dd{\xi^j}$ by $\dd{\eta^j}$, which commutes with $\ev$, such that
\begin{align*}
\ev|_{t=0}\Phi^*F[Y]&=(-1)^{(\abs{\eta^j}+\abs{\eta^i})\abs{\eta^i}}F(\partial_{\eta^i})\cdot\dd{\eta^j}\ev|_{t=0}\phi(\eta^i)\cdot Y^j\\
&=(-1)^{(\abs{\eta^j}+\abs{\eta^i})\abs{\eta^i}}F(\partial_{\eta^i})\cdot\dd{\eta^j}(\eta^i)\cdot Y^j\\
&=F(Y)
\end{align*}
\end{proof}

\begin{Lem}
\label{lemLieDerivativeConstruction}
Let $f\in\mO_M$ be a superfunction and $F,G\in\Hom_{\mO_M}(\mS M,\mO_M)$ be one forms.
Then
\begin{align*}
L_X(f\cdot F)&=L_Xf\cdot F+(-1)^{\abs{f}\abs{X}}f\cdot L_XF\\
L_X(F\otimes G)&=L_XF\otimes G+(-1)^{\abs{X}\abs{F}}F\otimes L_XG
\end{align*}
\end{Lem}

\begin{proof}
The first assertion follows from a straightforward calculation, using (\ref{eqnPullbackMultiplicationTensor}),
the derivation property of $D$ as well as $\ev|_{t=0}(f\cdot g)=(\ev|_{t=0} f)\cdot(\ev_{t=0} g)$
and Lem. \ref{lemLieDerivativeFunctions} and Lem. \ref{lemInitialCondition}.
The second assertion is shown analogous.
\end{proof}

Any multilinear form can be (locally) written as the tensor product of
one-forms of the form $df$ for a superfunction, multiplied with a superfunction.
By Lem. \ref{lemLieDerivativeFunctions} the Lie derivative of these building blocks
is independent of the integration model chosen (i.e. independent of $D$).
By Lem. \ref{lemLieDerivativeConstruction}, the Lie derivative of a general
multilinear form is uniquely determined by the building blocks and thus
also independent of the integration model. In particular, for an even bilinear form, one starts with
\begin{align*}
L_XB(Y,Z)=(-1)^{\abs{\xi^i}+\abs{\xi^j}+\abs{\xi^i}\abs{\xi^j}}L_X\left(B_{ij}\cdot d\xi^i\otimes d\xi^j\right)(Y,Z)
\end{align*}
as in (\ref{eqnLocalBilinear}). An explicit calculation yields the following result.

\begin{Lem}
\label{lemLieDerivativeBilinearForm}
Let $B\in\Hom_{\mO_M}(\mS M\otimes\mS M,\mO_M)$ be an even bilinear form. Then
\begin{align*}
L_X\scal[B]{Y}{Z}=X\scal[B]{Y}{Z}-\scal[B]{\scal[[]{X}{Y}}{Z}-(-1)^{\abs{X}\abs{Y}}\scal[B]{Y}{\scal[[]{X}{Z}}
\end{align*}
\end{Lem}

The right hand side of Lem. \ref{lemLieDerivativeBilinearForm} is taken
as the definition of the Lie derivative in \cite{Kli05}
(up to a global sign). For semi-Riemannian supermetrics, we have
the following characterisation, which naturally generalises the analogous
formula in the classical case. The proof is based on $\nabla$
being metric and torsion-free and (up to signs) the same as in the classical
case, thus omitted.

\begin{Lem}
\label{lemLieDerivativeLeviCivita}
Let $g$ be a semi-Riemannian supermetric and $\nabla$ the Levi-Civita superconnection.
Then the Lie derivative of $g$ can be written
\begin{align*}
L_X\scal[g]{Y}{Z}=(-1)^{\abs{X}\abs{Y}}\left(
\scal[g]{\nabla_YX}{Z}+(-1)^{\abs{X}\abs{Y}+\abs{X}\abs{Z}+\abs{Y}\abs{Z}}\scal[g]{\nabla_ZX}{Y}\right)
\end{align*}
\end{Lem}

\subsection{Killing Vector Fields}

\begin{Def}
Let $((M,\mO_M),g)$ be a semi-Riemannian supermanifold.
A \emph{Killing vector field} is a vector field $X$ such that $L_Xg=0$
\end{Def}

For the proof of the following characterisation theorem, we shall need the
flow equation for superbilinear forms as stated in the next lemma,
which is easily generalised to general multilinear forms.
A similar result is stated in \cite{MSV93} as Prp. 4.3.

\begin{Lem}
\label{lemBilinearFormFlowEquation}
Let $B$ be a superbilinear form. Then the flow equation holds:
\begin{align*}
\ev|_{t=t_0}D\circ\Phi^*B=\ev|_{t=t_0}\Phi^*L_XB
\end{align*}
Moreover, if $X$ possesses a strong flow s.th. $D\circ\phi=\phi\circ X$,
then the flow equation is satisfied without the $\ev$ morphism.
\end{Lem}

\begin{proof}
By a direct calculation along the lines of the proofs of Lem. \ref{lemLieDerivativeFunctions}
and Lem. \ref{lemInitialCondition}, one verifies the (weak) flow equations
\begin{align*}
\ev|_{t=t_0}D\circ\Phi^*f=\ev|_{t=t_0}\Phi^*L_Xf\;,\qquad
\ev|_{t=t_0}D\circ\Phi^*df=\ev|_{t=t_0}\Phi^*L_Xdf
\end{align*}
for a superfunction $f$ as well as $df$. Writing $B$ in local form as in (\ref{eqnLocalBilinear}),
a further calculation using these building blocks as well as Lem. \ref{lemLieDerivativeConstruction}
and (\ref{eqnPullbackMultiplicationTensor}) yields the assertion.
It is clear that all steps can be done without $\ev$ provided that
the strong flow condition holds.
\end{proof}

Let $g$ be a semi-Riemannian supermetric on $M$ with associated
$OSp_{(t,s)|2m}$-structure $\mF_g=\mF_{OSp_{(t,s)|2m}}$ as in Lem. \ref{lemSupermetricOSpStructure}.
Let $U\subseteq M$ be open. We write $E\in\mF_g(U)$ as $E=(X_1,\ldots,X_{t+s+2m})$ such that
$g(X_i,X_j)=g_0(e_i,e_j)$. Any supermatrix with entries in $\mO_M(U)$ acts on $E$
via (\ref{eqnRightAction}). Following \cite{ACDS97}, we call $U$ ''small'' if it is such that
the action of $OSp_{(t,s)|2m}(U)$ on $\mF_g(U)$ is simply transitive.
Moreover, we set $L_XE:=(\scal[[]{X}{X_1},\ldots,\scal[[]{X}{X_{t+s+2m}})$
for a vector field $X\in\mS M$. With this notation, our characterisation theorem can
be stated as follows.

\begin{Thm}
\label{thmKillingVF}
Let $X\in\mS M$ be a vector field. Then the following conditions are equivalent.
\begin{enumerate}
\renewcommand{\labelenumi}{(\roman{enumi})}
\item $X$ is Killing, i.e. $L_Xg=0$.
\item For all $Y,Z\in\mS M$, $\scal[g]{\nabla_YX}{Z}+(-1)^{\abs{X}\abs{Y}+\abs{X}\abs{Z}+\abs{Y}\abs{Z}}\scal[g]{\nabla_ZX}{Y}=0$.
\item The metric $\Phi^*g=g$ is preserved by the flow $\Phi$ of $X$.
\item $d\Phi\mF_G\subseteq\varphi_*^{-1}\mF_G$ for the flow $\Phi$ of $X$.
\item $L_XE\in E\cdot(\mathfrak{osp}_{(t,s)|2m}\otimes\mO_M(U))$ for all ''small'' $U\subseteq M$ and $E\in\mF_g(U)$.
\end{enumerate}
\end{Thm}

\begin{proof}
$(i)\iff(ii)$ is immediate by Lem. \ref{lemLieDerivativeLeviCivita}.

$(iii)\implies(i)$: Assume that $\Phi^*g=g$. Then, by $t$,$\tau$-indepence of $g$, we yield
\begin{align*}
L_Xg=\ev|_{t=0}D\circ\Phi^*g=\ev|_{t=0}D\circ g=0
\end{align*}

$(i)\implies(iii)$:
Let $X$ be a Killing vector field, we want to show that $\Phi^*g=g$ follows.
Let $Y,Z$ be vector fields of pure parity on $M$ and consider the superfunction
\begin{align*}
f:=\Phi^*g(Y,Z)\in\mO_{\mD(X)}\;,\qquad f=f^0+\tau\cdot f^1\;,\qquad f^0,f^1\in\mO_M
\end{align*}
Since $g$ is purely even and $\Phi$, being a morphism, is even, $f$ has a fixed parity
and thus $f^0$ and $f^1$ are of opposite parity (because of $\tau$).
By assumption, $L_Xg=0$ vanishes. Therefore, by the flow equation
of Lem. \ref{lemBilinearFormFlowEquation}, we have
\begin{align*}
0&=\ev|_{t=t_0}D(f)=\ev|_{t=t_0}(\partial_t+\partial_{\tau})(f^0+\tau\cdot f^1)
=\ev|_{t=t_0}(\partial_tf^0+\tau\cdot\partial_tf^1+f^1)\\
&=\partial_tf^0(t_0)+f^1(t_0)
\end{align*}
The first and second terms are of opposite parity (as stated above).
Therefore, each summand vanishes individually such that $f^1=0$
and $\partial_tf^0=0$. Consider the local representation $f^0=\sum_{J}f^0_J\theta^J$
with odd coordinates $\theta^i$ on $M$ and multiindices $J$, for which we yield
$\partial_tf^0_J=0$ (an equation for ordinary functions). Now, by Lem. \ref{lemInitialCondition},
we have the initial condition
\begin{align*}
f^0(t=0)=\ev|_{t=0}f=\ev|_{t=0}g(Y,Z)=g(Y,Z)
\end{align*}
where $f^0(t=0)=\sum_Jf^{\theta}_J(t=0)\theta^J$ (the right hand side can be
expanded into an analogous sum). Therefore, we conclude
\begin{align*}
\Phi^*g(Y,Z)=f=f^0=g(Y,Z)
\end{align*}
for all vector fields $Y,Z$ of fixed parity. Therefore, $\Phi^*g=g$
which was to be shown.

$(iii)\iff(iv)$: The proof of Lem. \ref{lemSupermetricAutomorphism} applies verbatim.

$(i)\iff(v)$:
Let $E=(X_1,\ldots)\in\mF_g(U)$ and $X\in\mS M$. There is a supermatrix $L\in\fgl_{t+s|2m}(\mO_M(U))$
such that $\scal[[]{X}{X_i}=X_m\cdot L_{mi}$. By Lem. \ref{lemLieDerivativeBilinearForm}, we have
\begin{align*}
L_Xg(X_i,X_j)&=-g(\scal[[]{X}{X_i},X_j)-(-1)^{\abs{X}\abs{X_i}}g(X_i,\scal[[]{X}{X_j})\\
&=-g(X_m\cdot L_{mi},X_j)-(-1)^{\abs{X}\abs{X_i}}g(X_i,X_m\cdot L_{mj})\\
&=-g_0(L\cdot e_i,e_j)-(-1)^{\abs{g}\abs{e_i}}g_0(e_i,L\cdot e_j)
\end{align*}
where $e_i$ is the standard basis of $\mO_M(U)^{t+s|2m}$ and $g_0$ is the standard supermetric
as in (\ref{eqnStandardMetric}).
It follows immediately that $L_Xg=0$ is equivalent to
$L\in\mathfrak{osp}_{(t,s)|2m}\otimes\mO_M(U)$.
\end{proof}

Finally, we consider Killing vector fields on spinor supermanifolds, the
example of supermanifolds as considered in \cite{ACDS97}. Let $(M,g)$ be a spin manifold
and consider a parallel non-degenerate suitable bulinear form $g_1$ on the spinor bundle $S$.
We also assume that $g_1$ is skew-symmetric (consult \cite{Har90} for a classification of such forms)
such that $g+g_1$ induces a Riemannian supermetric on the split supermanifold $(M,\Gamma(\bigwedge S))$.
There is a canonical monomorphism $\iota:\Gamma(TM)\oplus\Gamma(S^*)\rightarrow\mS M$ of sheaves which induces an isomorphism
$\iota:T_pM\oplus S^*_p\rightarrow S_pM$ for every $p\in M$.
It follows that $\{\iota(X_i),\iota(s^*_j)\}$ is a (local) basis for $\mS M$ if $\{X_i\}$ is a basis of $\Gamma(TM)$
and $\{s^*_j\}$ is a basis of $S^*$.
We identify sections $s^*\in\Gamma(S^*)$ with $s\in\Gamma(S)$ via $g_1$.

\begin{Lem}
\label{lemKillingSpinor}
Let $s^*\in\Gamma(S^*)$. Then the super vector field
$\iota(s^*)$ is Killing if and only if $s$ is a parallel spinor.
\end{Lem}

In particular, we obtain existence results for Killing vector fields on
spinor supermanifolds from the classification of spin manifolds admitting
parallel spinors \cite{MS00}. The definition of a Killing vector field in \cite{ACDS97}
is as in Thm. \ref{thmKillingVF}(v) but with $OSp_{(t,s)|2m}$ replaced by some
super Poincar\'e group. For the case of $\iota(s^*)$, this is proved to be equivalent
to $s$ being a twistor spinor, which is different from our characterisation.

\begin{proof}
It suffices to consider $L_{\iota(s^*)}g$ with vector fields from the image of $\iota$ inserted.

Let $t^*,u^*\in\Gamma(S^*)$ such that $\iota(t^*),\iota(u^*)\in\mS M_1$.
Using Lem. \ref{lemLieDerivativeBilinearForm} and Lem. 1 of \cite{ACDS97} we yield
$\scal[L_{\iota(s^*)}]{\iota(t^*)}{\iota(u^*)}=0$ which, therefore, is no condition.

Next, we consider $Y,Z\in\Gamma(TM)$ such that $\iota(Y),\iota(Z)\in\mS M_0$.
The type of supermetric considered here vanishes upon insertion of an even and an odd vector field.
We are thus led to $\scal[L_{\iota(s^*)}]{\iota(Y)}{\iota(Z)}=0$ which is again no condition.

Next, we find $\scal[L_{\iota(s^*)}]{\iota(t^*)}{\iota(Y)}=-\scal[g_1]{t}{\nabla_Ys}$.
A necessary and sufficient condition for the last term to vanish for all $t^*$ and $Y$ is
$\nabla s\equiv 0$.
\end{proof}

\section{Superharmonic Field Theories}
\label{secSuperharmonic}

The classical harmonic action functional for maps $\varphi:M\rightarrow N$ between semi-Riemannian manifolds
$(M,h)$ and $(N,g)$ reads
\begin{align}
\label{eqnHarmonicAction}
\mA(\varphi)=\frac{1}{2}\int_M\dvol_h\,\tr_h(\varphi^*g)
=\frac{1}{2}\int_M\dvol_h h^{ij}(x)\scal[(\varphi^*g)_x]{\dd{x^i}}{\dd{x^j}}
\end{align}
Critical points of this functional are called \emph{harmonic maps}.
For an exhaustive treatment of the Riemannian case, consult \cite{Xin96}.
In this section, we study a natural generalisation of (\ref{eqnHarmonicAction}) for
semi-Riemannian supermanifolds $(M,h)$ and $(N,g)$ and prove three Noether theorems in this context.

Throughout, $M$ is assumed to be compact and superoriented as explained below after introducing some terminology.
For a superfunction $f\in\mO_M$, we write $f>0$ if $\tilde{f}(p)>0$ for all $p$.
Moreover, we set $\abs{f}:=m\cdot f$, where $m:M\rightarrow\{\pm 1\}$ is such that $\abs{f}>0$.
It can be shown that, in case $f$ is even and $f>0$, there is a unique square root
$\sqrt{f}$ which is constructed using Taylor-like expansion in odd coordinates.
Now, $M$ is called superoriented if it has an atlas of coordinate charts such that,
for every coordinate transformation $\Phi=(\varphi,\phi):\bR^{n|m}\rightarrow\bR^{n|m}$,
both $\det(d\varphi)>0$ and $\sdet(d\Phi)>0$ hold \cite{Sha88}.
Here, the first condition (classical orientedness) is needed to make the integral over
densities (sections of the superdeterminant sheaf, see Chp. 3 of \cite{DM99}) welldefined.
In case of the second condition, the metric $h$ induces a canonical volume form (density)
$\dsvol_h$ on $M$ as follows. Let $\xi=(x,\theta)$ be local coordinates on $M$ and denote by
$[d^nxd^m\theta]$ the induced local density. Moreover, let $\sdet(h)$ denote the super-determinant of the matrix
$\scal[h]{\partial_{\xi^l}}{\partial_{\xi^k}}$. With this notation, $\dsvol_h$ can be defined by the local expression
\begin{align*}
\dsvol_h=[d^nxd^m\theta]\cdot\sqrt{\abs{\sdet h}}
\end{align*}
which is independent of the coordinates used. A treatment of Riemannian volume forms in a slightly different style
can be found in Sec. 3.5 of \cite{Han12}.

Instead of plain morphisms $\Phi=(\varphi,\phi):M\rightarrow N$ of supermanifolds, we consider
morphisms $\Phi:M\times\bR^{0|L}\rightarrow N$ in order to obtain odd ''component fields'' (for simplicity,
think of terms of $\phi$) as models for fermions. Following \cite{Hel09}, we call such morphisms \emph{maps with flesh},
while the same concept occurs with several names in the literature, see \cite{DF99b} and \cite{Khe07}.
For the following treatment, it suffices to consider a fixed value of $L\in\bN$ which is large enough
to make the calculations consistent, cf. the discussion in \cite{Hel09}. For a functorial point of view
(considering all values of $L$ simultaneously), we refer to \cite{Han12}.
In the following, we shall simply write $\Phi:M\rightarrow N$ for maps with flesh, leaving the
superpoint $\bR^{0|L}$ implicit. The differential $d\Phi$ is then implicitly rescricted to the
tangent sheaf of $M$ tensored with the algebra $\bigwedge\bR^L$ of superfunctions of $\bR^{0|L}$.
Tensors on $\mS M$ are similarly endowed to that sheaf by $\bigwedge\bR^L$-multilinear extension.
For details, consult \cite{Gro11a}.

With these preparations, we define the \emph{superharmonic action functional} as
\begin{align}
\label{eqnSuperAction}
\mA(\Phi):=\frac{1}{2}\int_M \dsvol_h\,\str_h(\Phi^*g)=\frac{1}{2}\int_M \dsvol_h\scal[g_{\Phi}]{d\Phi[e_j]}{d\Phi[Je_j]}
\end{align}
with $\str_h$ as in (\ref{eqnStraceMetric}), and where $\{e_j\}$ is a local $OSp_{(t,s)|2m}$-frame
on $(M,h)$.

\begin{Def}
\label{defSecondFundForm}
Let $\Phi:(M,h)\rightarrow(N,g)$ be a morphism. Then
\begin{align*}
B_{X,Y}(\Phi):=(\nabla_Xd\Phi)[Y]=\nabla_X(d\Phi[Y])-d\Phi[\nabla_XY]
\end{align*}
is called \emph{2nd fundamental form}.
\end{Def}

\begin{Lem}
\label{lemSecondFundForm}
The 2nd fundamental form is a tensor $B_{\cdot,\cdot}(\Phi)\in\Hom(\mS M\otimes_{\mO_M}\mS M,\mS\Phi)$.
Moreover, it is supersymmetric, i.e. $B_{X,Y}=(-1)^{\abs{X}\abs{Y}}B_{Y,X}$.
\end{Lem}

\begin{proof}
$B$ is a tensor since $B_{fX,Y}(\Phi)=(-1)^{\abs{X}\abs{f}}B_{X,fY}(\Phi)=fB_{X,Y}(\Phi)$ is satisfied.
To show supersymmetry, it thus suffices to consider coordinate vector fields.
One verifies that the expression of $B_{\partial_{\xi^i},\partial_{\xi^j}}(\Phi)$ in terms of
coordinates $\{\xi^i\}$ on $M$ is supersymmetric in $i\leftrightarrow j$.
\end{proof}

\begin{Def}
\label{defTensionField}
The super trace of the second fundamental form is called \emph{tension field}.
\begin{align*}
\tau(\Phi):=\str_hB=\str_h(\nabla_{\cdot}d\Phi)[\cdot]=(\nabla_{e_j}d\Phi)[Je_j]
\end{align*}
\end{Def}

\begin{Thm}[\cite{Han12}, Thm. 6.29]
\label{thmSuperEulerLagrange}
$\Phi$ is a critical point of the action functional (\ref{eqnSuperAction}) if and only if the Euler-Lagrange equation $\tau(\Phi)=0$ holds.
\end{Thm}

The corresponding theorem in \cite{Han12} is formulated for a special case (in particular, $h$ is Riemannian there),
but the proof provided there applies to the general case.

We also need the divergence of a super vector field. Classically, it can be defined as $\tr(\nabla X)$.
In the super case, an additional sign occurs since the map $Y\mapsto\nabla_YX$ (which has the parity of $X$)
is not a superlinear map in case $X$ is odd.

\begin{Def}
\label{defDivergence}
We set $\div X:=\str\left(Y\mapsto(-1)^{\abs{X}\abs{Y}}\nabla_YX\right)=(-1)^{\abs{e_j}\abs{X}}\scal[g]{\nabla_{e_j}X}{Je_j}$.
\end{Def}

For the characterisation in a local $OSp_{(t,s)|2m}$-frame, beware that the supertrace
\begin{align*}
\div X=(-1)^{\abs{e_j}(\abs{X}+1)}\left(Y\mapsto(-1)^{\abs{X}\abs{Y}}\nabla_YX\right)^j_{\phantom{j}j}
\end{align*}
is defined with respect to right coordinates.
We define next the analogon for a vector field along $\Phi$. Consider the super-bilinear form
$(X,Y)\mapsto(-1)^{\abs{X}\abs{\xi}}\scal[g_{\Phi}]{\nabla_X\xi}{d\Phi[Y]}$. Again, the sign
is necessary to make it a super-bilinear form.

\begin{Def}
\label{defDivergenceAlongPhi}
Let $\xi\in\mS\Phi$ be a vector field along $\Phi$.
We define its divergence to be
\begin{align*}
\div\xi:=\str_h\left((X,Y)\mapsto(-1)^{\abs{X}\abs{\xi}}\scal[g_{\Phi}]{\nabla_{X}\xi}{d\Phi[Y]}\right)
=(-1)^{\abs{e_i}\abs{\xi}}\scal[g_{\Phi}]{\nabla_{e_i}\xi}{d\Phi[Je_i]}
\end{align*}
\end{Def}

\begin{Lem}
\label{lemDivAlongMorphism}
Let $\xi\in\mS\Phi$ and set $W_{\xi}:=\scal[g_{\Phi}]{\xi}{d\Phi[e_j]}Je_j$ which has the parity of $\xi$. Then
$\div W_{\xi}=\div\xi+\scal[g_{\Phi}]{\xi}{\tau(\Phi)}$.
\end{Lem}

\begin{proof}
The assertion is shown by the following calculation, using $\abs{W_{\xi}}=\abs{\xi}$ as
well as the metric property of both $h$ and $g_{\Phi}$.
\begin{align*}
\div W_{\xi}
&=(-1)^{\abs{e_i}(1+\abs{\xi})}\scal[h]{\nabla_{Je_i}[\scal[g_{\Phi}]{\xi}{d\Phi[e_j]}Je_j]}{e_i}\\
&=(-1)^{\abs{e_i}(1+\abs{\xi})}Je_i\scal[g_{\Phi}]{\xi}{d\Phi[e_j]}\cdot\scal[h]{Je_j}{e_i}\\
&\qquad+(-1)^{\abs{e_i}\abs{e_j}+\abs{e_i}}\scal[g_{\Phi}]{\xi}{d\Phi[e_j]}\scal[h]{\nabla_{Je_i}(Je_j)}{e_i}\\
&=(-1)^{\abs{e_i}(1+\abs{\xi})}Je_j\scal[g_{\Phi}]{\xi}{d\Phi[e_j]}
-(-1)^{\abs{e_i}}\scal[g_{\Phi}]{\xi}{d\Phi[e_j]}\scal[h]{Je_j}{\nabla_{Je_i}e_i}\\
&=(-1)^{\abs{e_j}(1+\abs{\xi})}Je_j\scal[g_{\Phi}]{\xi}{d\Phi[e_j]}
-(-1)^{\abs{e_i}}\scal[g_{\Phi}]{\xi}{d\Phi[\nabla_{Je_i}e_i]}\\
&=(-1)^{\abs{e_j}(1+\abs{\xi})}\scal[g_{\Phi}]{\nabla_{Je_j}\xi}{d\Phi[e_j]}
+(-1)^{\abs{e_j}}\scal[g_{\Phi}]{\xi}{\nabla_{Je_j}d\Phi[e_j]-d\Phi[\nabla_{Je_j}e_j]}\\
&=\div\xi+\scal[g_{\Phi}]{\xi}{\tau(\Phi)}
\end{align*}
\end{proof}

\subsection{Target Space Symmetries}

Killing vector fields $\xi\in\Gamma(TN)$ on the target space are infinitesimal symmetries
of the harmonic action (\ref{eqnHarmonicAction}). This can be seen as follows. Consider the flow
$F_t$ of $\xi$. We alter $\varphi$ by moving along the flow lines via $\varphi_t(x):=F_t(\varphi(x))$.
In physicists' notation, this means that, infinitesimally, $\varphi(x)\rightarrow\varphi(x)+\xi(\varphi(x))dt$.
The infinitesimal change of $\mA(\varphi)$ by $\xi$ thus becomes
\begin{align*}
\frac{d}{dt}|_0\mA(F_t\circ\varphi)&=\frac{1}{2}\frac{d}{dt}|_0\int_M\dvol_h\,\tr_h((F_t\circ\varphi)^*g)
=\frac{1}{2}\int_M\dvol_h\,\tr_h\left(\varphi^*\dd{t}|_0F_t^*g\right)\\
&=\frac{1}{2}\int_M\dvol_h\,\tr_h(\varphi^*L_{\xi}g)
\end{align*}

Consider now the context of the superharmonic action (\ref{eqnSuperAction}).
We let $\xi\in\mS N$ be a vector field on $N$ and denote by $F:\mD(\xi)\times N\rightarrow N$ its flow
as in Def. \ref{defFlow}. Note that $F$ is a plain morphism of supermanifolds
while $\Phi$ is a map with flesh.
The analogon of $F_t(\varphi(x))$ is $F\circ\Phi:=F\circ(\id\times\Phi):\mD(\xi)\times M\rightarrow N$,
and the (finite) change of $\mA(\Phi)$ by $F$ reads
\begin{align*}
\mA(F\circ\Phi)=\frac{1}{2}\int_M\dsvol_h\,\str_h\left((F\circ\Phi)^*g\right)
=\frac{1}{2}\int_M\dsvol_h\scal[(F\circ\Phi)^*g]{e_j}{Je_j}
\end{align*}

\begin{Dlm}
\label{lemTargetSpaceChange}
The infinitesimal change of $\mA$ by $\xi\in\mS N$ is
\begin{align*}
\ev|_{t=0}D\mA(F\circ\Phi)=\frac{1}{2}\int_m\dsvol_h\,\str_h\left(\Phi^*(L_{\xi}g)\right)
\end{align*}
\end{Dlm}

\begin{proof}
By compactness of $M$, we may interchange integration and differentiation, such that
\begin{align*}
\ev|_{t=0}D\mA(F\circ\Phi)
&=\frac{1}{2}\int_M\dsvol_h\ev_{t=0}D\scal[((F\circ\Phi)^*g)]{e_j}{Je_j}\\
&=\frac{1}{2}\int_M\dsvol_h\ev_{t=0}D\scal[(\Phi^*F^*g)]{e_j}{Je_j}\\
&=\frac{1}{2}\int_M\dsvol_h\Phi^*\ev_{t=0}D\scal[(F^*g)]{e_j}{Je_j}\\
&=\frac{1}{2}\int_M\dsvol_h\Phi^*\scal[(L_{\xi}g)]{e_j}{Je_j}\\
&=\frac{1}{2}\int_m\dsvol_h\,\str_h\left(\Phi^*(L_{\xi}g)\right)
\end{align*}
In this calculation, the third equation holds since only $F^*g$ depends on the flow coordinates
on $\mD(\xi)$. It is proved by a straightforward calculation in local coordinates.
\end{proof}

It follows that, again, Killing vector fields on $N$ are infinitesimal symmetries!
According to the Noether principle, there should be an induced conserved quantity.
We will show next that this is indeed the case. The need the following analogon of
Lem. \ref{lemLieDerivativeLeviCivita}.

\begin{Lem}
\label{lemPullbackLieDerivative}
Let $\nabla=\nabla_{\Phi}$ denote the pullback connection as in (\ref{eqnPullbackConnection}). Then
\begin{align*}
&\scal[\Phi^*(L_{\xi}g)]{Y}{Z}\\
&\qquad=(-1)^{\abs{\xi}\abs{Y}}\scal[g_{\Phi}]{\nabla_Y(\phi\circ\xi)}{d\Phi[Z]}
+(-1)^{\abs{\xi}\abs{Y}+\abs{\xi}\abs{Z}}\scal[g_{\Phi}]{d\Phi[Y]}{\nabla_Z(\phi\circ\xi)}
\end{align*}
\end{Lem}

\begin{proof}
Choosing local coordinates $\{i=\eta^i\}$ on $N$, the assertion is reduced to Lem. \ref{lemLieDerivativeLeviCivita} as follows.
\begin{align*}
&\scal[\Phi^*(L_{\xi}g)]{Y}{Z}\\
&\qquad=\scal[(L_{\xi}g)_{\Phi}]{(\phi\circ\partial_i)d\Phi[Y]^i}{(\phi\circ\partial_j)d\Phi[Z]^j}\\
&\qquad=(-1)^{(\abs{i}+\abs{Y})(\abs{i}+\abs{\xi})}d\Phi[Y]^i\phi\circ\scal[L_{\xi}g]{\partial_i}{\partial_j}\cdot d\Phi[Z]^j\\
&\qquad=(-1)^{\abs{i}+\abs{i}\abs{Y}+\abs{Y}\abs{\xi}}d\Phi[Y]^i\cdot\\
&\qquad\qquad\qquad\qquad\phi\circ\left(\scal[g]{\nabla_{\partial_i}\xi}{\partial_j}
+(-1)^{\abs{\xi}\abs{i}+\abs{\xi}\abs{j}+\abs{i}\abs{j}}\scal[g]{\nabla_{\partial_j}\xi}{\partial_i}\right)\cdot d\Phi[Z]^j\\
&\qquad=(-1)^{\abs{\xi}\abs{Y}}\scal[g_{\Phi}]{\nabla_Y(\phi\circ\xi)}{d\Phi[Z]}
+(-1)^{\abs{\xi}\abs{Y}+\abs{\xi}\abs{Z}}\scal[g_{\Phi}]{d\Phi[Y]}{\nabla_Z(\phi\circ\xi)}
\end{align*}
\end{proof}

\begin{Thm}[Noether]
\label{thmNoetherTargetSymmetry}
Let $\xi\in\mS N$ be a Killing vector field ($L_{\xi}g=0$). Then the divergence $\div(\phi\circ\xi)=0$ vanishes.
If, moreover, $\Phi$ is a superharmonic map (solution of the Euler-Lagrange equation $\tau(\Phi)=0$), then
$\div W_{\phi\circ\xi}=0$ vanishes, too, where $W_{\phi\circ\xi}:=\scal[g_{\Phi}]{\phi\circ\xi}{d\Phi[e_j]}Je_j$.
\end{Thm}

\begin{proof}
$L_{\xi}g=0$ implies
\begin{align*}
0&=\scal[\Phi^*(L_{\xi}g)]{e_j}{Je_j}\\
&=(-1)^{\abs{\xi}\abs{e_j}}\scal[g_{\Phi}]{\nabla_{e_j}(\phi\circ\xi)}{d\Phi[Je_j]}+\scal[g_{\Phi}]{d\Phi[e_j]}{\nabla_{Je_j}(\phi\circ\xi)}\\
&=(-1)^{\abs{\xi}\abs{e_j}}\left(\scal[g_{\Phi}]{\nabla_{e_j}(\phi\circ\xi)}{d\Phi[Je_j]}+(-1)^{\abs{e_j}}\scal[g_{\Phi}]{\nabla_{Je_j}(\phi\circ\xi)}{d\Phi[e_j]}\right)\\
&=2\div(\phi\circ\xi)
\end{align*}
using Lem. \ref{lemPullbackLieDerivative}.
The second statement now follows immediately from Lem. \ref{lemDivAlongMorphism}.
\end{proof}

\subsection{Domain Space Symmetries}

We have seen that the infinitesimal change of the superharmonic action (\ref{eqnSuperAction})
by a Killing vector field on the target space vanishes. Let us now consider the corresponding
infinitesimal change by a vector field $\xi\in\mS M$ on the domain space with flow
$F:\mD(\xi)\times M\rightarrow M$.

\begin{Dlm}
The infinitesimal change of $\mA$ by $\xi\in\mS M$ is
\begin{align*}
\ev|_{t=0}D\mA(\Phi\circ F)=\frac{1}{2}\int_M\dsvol_h\,\str(L_{\xi}(\Phi^*g))
\end{align*}
\end{Dlm}

\begin{proof}
Analogous to the proof of Lem. \ref{lemTargetSpaceChange}, we calculate
\begin{align*}
\ev|_{t=0}D\mA(\Phi\circ F)
&=\frac{1}{2}\int_M\dsvol_h\ev_{t=0}D\scal[((\Phi\circ F)^*g)]{e_j}{Je_j}\\
&=\frac{1}{2}\int_M\dsvol_h\ev_{t=0}DF^*\scal[(\Phi^*g)]{e_j}{Je_j}\\
&=\frac{1}{2}\int_M\dsvol_h\scal[L_{\xi}(\Phi^*g)]{e_j}{Je_j}\\
&=\frac{1}{2}\int_M\dsvol_h\,\str(L_{\xi}(\Phi^*g))
\end{align*}
\end{proof}

As usual, $\xi$ is called an infinitesimal symmetry if this expression vanishes for all morphisms $\Phi$.
Opposed to the target space situation, the vanishing of the Lie derivative $L_{\xi}(\Phi^*g)$ (for all $\Phi$)
is harder to achieve. In case symmetry is present, it is usually only such that $L_{\xi}(\Phi^*g)$ is some exact term
depending on $\Phi$ (but integrated over to zero).

\begin{Thm}[Noether]
\label{thmNoetherDomain}
Let $\xi\in\mS M$ be a $\Phi$-Killing vector field, i.e. such
that $L_{\xi}(\Phi^*g)=0$. Then $\div(d\Phi[\xi])=\div(\xi\circ\phi)=0$ vanishes.
If, moreover, $\Phi$ is a superharmonic map, then $\div W_{d\Phi[\xi]}=0$ vanishes,
where $W_{d\Phi[\xi]}=\scal[\Phi^*g]{\xi}{e_j}Je_j$.
\end{Thm}

\begin{proof}
Here, $\Phi^*g$ is a supermetric. However, Lem. \ref{lemLieDerivativeLeviCivita} is not directly applicable
since it would lead to the Levi-Civita connection $\nabla^{\Phi^*g}$ of this supermetric rather than to
the Levi-Civita connection $\nabla=\nabla^h$ of $h$. We thus step back and use Lem. \ref{lemLieDerivativeBilinearForm}
as well as torsion-freeness of $\nabla$ to obtain
\begin{align*}
0&=\scal[L_{\xi}(\Phi^*g)]{e_j}{Je_j}\\
&=\xi\scal[\Phi^*g]{e_j}{Je_j}-\scal[\Phi^*g]{\scal[[]{\xi}{e_j}}{Je_j}-(-1)^{\abs{\xi}\abs{e_j}}\scal[\Phi^*g]{e_j}{\scal[[]{\xi}{Je_j}}\\
&=\xi\scal[\Phi^*g]{e_j}{Je_j}-\scal[\Phi^*g]{\nabla_{\xi}e_j}{Je_j}+(-1)^{\abs{\xi}\abs{e_j}}\scal[\Phi^*g]{\nabla_{e_j}\xi}{Je_j}\\
&\qquad-(-1)^{\abs{\xi}\abs{e_j}}\scal[\Phi^*g]{e_j}{\nabla_{\xi}(Je_j)}+\scal[\Phi^*g]{e_j}{\nabla_{Je_j}\xi}
\end{align*}
Now, by Lem. \ref{lemMetricPullback}, $g_{\Phi}$ is metric, and we thus obtain
\begin{align*}
0&=\scal[g_{\Phi}]{\nabla_{\xi}d\Phi[e_j]}{d\Phi[Je_j]}+(-1)^{\abs{e_j}\abs{\xi}}\scal[g_{\Phi}]{d\Phi[e_j]}{\nabla_{\xi}d\Phi[Je_j]}\\
&\qquad-\scal[g_{\Phi}]{d\Phi[\nabla_{\xi}e_j]}{d\Phi[Je_j]}
+(-1)^{\abs{\xi}\abs{e_j}}\scal[g_{\Phi}]{d\Phi[\nabla_{e_j}\xi]}{d\Phi[Je_j]}\\
&\qquad-(-1)^{\abs{\xi}\abs{e_j}}\scal[g_{\Phi}]{d\Phi[e_j]}{d\Phi[\nabla_{\xi}(Je_j)]}
+\scal[g_{\Phi}]{d\Phi[e_j]}{d\Phi[\nabla_{Je_j}\xi]}
\end{align*}
We combine the first and third and the second and fifth terms, respectively, such that
\begin{align*}
0&=\scal[g_{\Phi}]{(\nabla_{\xi}d\Phi)[e_j]}{d\Phi[Je_j]}
+(-1)^{\abs{e_j}\abs{\xi}}\scal[g_{\Phi}]{d\Phi[e_j]}{(\nabla_{\xi}d\Phi)[Je_j]}\\
&\qquad+(-1)^{\abs{\xi}\abs{e_j}}\left(\scal[g_{\Phi}]{\nabla_{e_j}d\Phi[\xi]}{d\Phi[Je_j]}
-\scal[g_{\Phi}]{(\nabla_{e_j}d\Phi)[\xi]}{d\Phi[Je_j]}\right)\\
&\qquad+\left(\scal[g_{\Phi}]{d\Phi[e_j]}{\nabla_{Je_j}d\Phi[\xi]}
-\scal[g_{\Phi}]{d\Phi[e_j]}{(\nabla_{Je_j}d\Phi)[\xi]}\right)
\end{align*}
We combine the first and fourth and second and sixth terms, respectively, and use Lem. \ref{lemSecondFundForm} such that
\begin{align*}
0&=\scal[g_{\Phi}]{B_{\xi,e_j}(\Phi)-(-1)^{\abs{\xi}\abs{e_j}}B_{e_j,\xi}(\Phi)}{d\Phi[Je_j]}\\
&\qquad+\scal[g_{\Phi}]{d\Phi[e_j]}{-B_{Je_j,\xi}(\Phi)+(-1)^{\abs{e_j}\abs{\xi}}B_{\xi,Je_j}(\Phi)}\\
&\qquad+\scal[g_{\Phi}]{d\Phi[e_j]}{\nabla_{Je_j}d\Phi[\xi]}+(-1)^{\abs{\xi}\abs{e_j}}\scal[g_{\Phi}]{\nabla_{e_j}d\Phi[\xi]}{d\Phi[Je_j]}\\
&=\scal[g_{\Phi}]{d\Phi[e_j]}{\nabla_{Je_j}d\Phi[\xi]}+(-1)^{\abs{\xi}\abs{e_j}}\scal[g_{\Phi}]{\nabla_{e_j}d\Phi[\xi]}{d\Phi[Je_j]}\\
&=(-1)^{\abs{e_j}+\abs{\xi}\abs{e_j}}\scal[g_{\Phi}]{\nabla_{Je_j}d\Phi[\xi]}{d\Phi[e_j]}
+(-1)^{\abs{\xi}\abs{e_j}}\scal[g_{\Phi}]{\nabla_{e_j}d\Phi[\xi]}{d\Phi[Je_j]}\\
&=2\div(d\Phi[\xi])
\end{align*}
The second statement now follows immediately from Lem. \ref{lemDivAlongMorphism}.
\end{proof}

\subsection{Domain Space Symmetries II}

In the previous subsection, we have considered $\Phi$-Killing vector fields $\xi\in\mS M$. We will next
prove a Noether theorem for the more common vector fields $\xi\in\mS M$ which are Killing
with respect to the metric $h$, thus generalising a classical result due to Baird and Eells \cite{BE81}.
We denote the super energy of $\Phi$ by
\begin{align*}
e(\Phi):=\frac{1}{2}\str_h\Phi^*g=\frac{1}{2}\scal[\Phi^*g]{e_j}{Je_j}
\end{align*}
and define the stress-energy tensor by
\begin{align*}
S_{\Phi}:=e(\Phi)h-\Phi^*g\in\Hom_{\mO_M}(\mS M\otimes_{\mO_M}\mS M,\mO_M)
\end{align*}
Moreover, for any tensor $S\in\Hom_{\mO_M}(\mS M\otimes_{\mO_M}\mS M,\mO_M)$, we define
\begin{align*}
\scal[(\nabla_XS)]{Y}{Z}:=X\scal[S_{\Phi}]{Y}{Z}-\scal[S]{\nabla_XY}{Z}-(-1)^{\abs{X}\abs{Y}}\scal[S]{Y}{\nabla_XZ}
\end{align*}
and
\begin{align*}
\div S[\xi]:=\str_h\left((X,Z)\mapsto(-1)^{\abs{X}\abs{\xi}}\scal[(\nabla_XS)]{\xi}{Z}\right)
=(-1)^{\abs{e_i}\abs{\xi}}\scal[(\nabla_{e_i}S)]{\xi}{Je_i}
\end{align*}
where $\xi\in\mS M$ is a vector field. As for the signs, cf. the discussion in the context of Def. \ref{defDivergence}.

\begin{Lem}
\label{lemStressEnergy}
Let $\xi\in\mS M$. Then
\begin{align*}
\div S_{\Phi}[\xi]=-\scal[g_{\Phi}]{d\Phi[\xi]}{\tau(\Phi)}
\end{align*}
\end{Lem}

\begin{proof}
We calculate
\begin{align*}
\div S_{\Phi}[\xi]
&=(-1)^{\abs{e_i}\abs{\xi}}e_i\scal[S_{\Phi}]{\xi}{Je_i}-(-1)^{\abs{e_i}\abs{\xi}}\scal[S_{\Phi}]{\nabla_{e_i}\xi}{Je_i}-\scal[S_{\Phi}]{\xi}{\nabla_{e_i}(Je_i)}\\
&=(-1)^{\abs{e_i}\abs{\xi}}e_i\left(\frac{1}{2}\scal[g_{\Phi}]{d\Phi[e_j]}{d\Phi[Je_j]}\scal[h]{\xi}{Je_i}-\scal[g_{\Phi}]{d\Phi[\xi]}{d\Phi[Je_i]}\right)\\
&\qquad-(-1)^{\abs{e_i}\abs{\xi}}e(\Phi)\scal[h]{\nabla_{e_i}\xi}{Je_i}-e(\Phi)\scal[h]{\xi}{\nabla_{e_i}(Je_i)}\\
&\qquad+(-1)^{\abs{e_i}\abs{\xi}}\scal[g_{\Phi}]{d\Phi[\nabla_{e_i}\xi]}{d\Phi[Je_i]}+\scal[g_{\Phi}]{d\Phi[\xi]}{d\Phi[\nabla_{e_i}(Je_i)]}\\
&=\frac{1}{2}\xi\circ\scal[g_{\Phi}]{d\Phi[e_j]}{d\Phi[Je_j]}+(-1)^{\abs{e_i}\abs{\xi}}e(\Phi)e_i\circ\scal[h]{\xi}{Je_i}\\
&\qquad-(-1)^{\abs{e_i}\abs{\xi}}e_i\circ\scal[g_{\Phi}]{d\Phi[\xi]}{d\Phi[Je_i]}\\
&\qquad-(-1)^{\abs{e_i}\abs{\xi}}e(\Phi)\scal[h]{\nabla_{e_i}\xi}{Je_i}-e(\Phi)\scal[h]{\xi}{\nabla_{e_i}(Je_i)}\\
&\qquad+(-1)^{\abs{e_i}\abs{\xi}}\scal[g_{\Phi}]{d\Phi[\nabla_{e_i}\xi]}{d\Phi[Je_i]}+\scal[g_{\Phi}]{d\Phi[\xi]}{d\Phi[\nabla_{e_i}(Je_i)]}
\end{align*}
Here, the second term cancels with the fourth and fifth such that
\begin{align*}
\div S_{\Phi}[\xi]
&=\frac{1}{2}\xi\circ\scal[g_{\Phi}]{d\Phi[e_j]}{d\Phi[Je_j]}-(-1)^{\abs{e_i}\abs{\xi}}e_i\circ\scal[g_{\Phi}]{d\Phi[\xi]}{d\Phi[Je_i]}\\
&\qquad+(-1)^{\abs{e_i}\abs{\xi}}\scal[g_{\Phi}]{d\Phi[\nabla_{e_i}\xi]}{d\Phi[Je_i]}+\scal[g_{\Phi}]{d\Phi[\xi]}{d\Phi[\nabla_{e_i}(Je_i)]}\\
&=\frac{1}{2}\scal[g_{\Phi}]{\nabla_{\xi}d\Phi[e_j]}{d\Phi[Je_j]}+\frac{1}{2}(-1)^{\abs{e_j}\abs{\xi}}\scal[g_{\Phi}]{d\Phi[e_j]}{\nabla_{\xi}d\Phi[Je_j]}\\
&\qquad-(-1)^{\abs{e_i}\abs{\xi}}\scal[g_{\Phi}]{\nabla_{e_i}d\Phi[\xi]}{d\Phi[Je_i]}-\scal[g_{\Phi}]{d\Phi[\xi]}{\nabla_{e_i}d\Phi[Je_i]}\\
&\qquad+(-1)^{\abs{e_i}\abs{\xi}}\scal[g_{\Phi}]{d\Phi[\nabla_{e_i}\xi]}{d\Phi[Je_i]}+\scal[g_{\Phi}]{d\Phi[\xi]}{d\Phi[\nabla_{e_i}(Je_i)]}
\end{align*}
By supersymmetry of $g_{\Phi}$ and $h$, we see that the first two terms coincide.
Moreover, we combine the third and fifth and the fourth and sixth terms, respectively, such that
\begin{align*}
\div S_{\Phi}[\xi]
&=\scal[g_{\Phi}]{\nabla_{\xi}d\Phi[e_j]}{d\Phi[Je_j]}-(-1)^{\abs{e_i}\abs{\xi}}\scal[g_{\Phi}]{(\nabla_{e_i}d\Phi)[\xi]}{d\Phi[Je_i]}\\
&\qquad-\scal[g_{\Phi}]{d\Phi[\xi]}{(\nabla_{e_i}d\Phi)[Je_i]}\\
&=\scal[g_{\Phi}]{(\nabla_{\xi}d\Phi)[e_j]-(-1)^{\abs{e_j}\abs{\xi}}(\nabla_{e_j}d\Phi)[\xi]+d\Phi[\nabla_{\xi}e_j]}{d\Phi[Je_j]}\\
&\qquad-\scal[g_{\Phi}]{d\Phi[\xi]}{\tau(\Phi)}\\
&=\scal[g_{\Phi}]{d\Phi[\nabla_{\xi}e_j]}{d\Phi[Je_j]}-\scal[g_{\Phi}]{d\Phi[\xi]}{\tau(\Phi)}
\end{align*}
using Lem. \ref{lemSecondFundForm}. Here, the first term vanishes by symmetry considerations.
\end{proof}

\begin{Lem}
\label{lemDivY}
Let $\xi\in\mS M$. Then, for $Y_{\xi}:=\scal[S_{\Phi}]{\xi}{e_i}Je_i$, we have
\begin{align*}
\div Y_{\xi}=\div S_{\Phi}[\xi]+\frac{1}{2}(-1)^{\abs{e_j}}\scal[L_{\xi}h]{e_i}{Je_j}\scal[S_{\Phi}]{e_j}{Je_i}
\end{align*}
\end{Lem}

\begin{proof}
We calculate
\begin{align*}
\div Y_{\xi}
&=(-1)^{\abs{e_j}(1+\abs{\xi})}\scal[h]{\nabla_{Je_j}[\scal[S_{\Phi}]{\xi}{e_i}Je_i]}{e_j}\\
&=(-1)^{\abs{e_j}(1+\abs{\xi})}Je_j\scal[S_{\Phi}]{\xi}{e_i}\cdot\scal[h]{Je_i}{e_j}\\
&\qquad\qquad\qquad+(-1)^{\abs{e_j}+\abs{e_j}\abs{e_i}}\scal[S_{\Phi}]{\xi}{e_i}\scal[h]{\nabla_{Je_j}(Je_i)}{e_j}\\
&=(-1)^{\abs{e_i}(1+\abs{\xi})}Je_i\scal[S_{\Phi}]{\xi}{e_i}
+(-1)^{\abs{e_j}+\abs{e_j}\abs{e_i}}\scal[S_{\Phi}]{\xi}{e_i}\scal[h]{\nabla_{Je_j}(Je_i)}{e_j}\\
&=(-1)^{\abs{e_i}(1+\abs{\xi})}\scal[(\nabla_{Je_i}S_{\Phi})]{\xi}{e_i}
+(-1)^{\abs{e_i}(1+\abs{\xi})}\scal[S_{\Phi}]{\nabla_{Je_i}\xi}{e_i}\\
&\qquad+(-1)^{\abs{e_i}}\scal[S_{\Phi}]{\xi}{\nabla_{Je_i}e_i}
-(-1)^{\abs{e_j}}\scal[S_{\Phi}]{\xi}{e_i}\scal[h]{Je_i}{\nabla_{Je_j}e_j}\\
&=(-1)^{\abs{e_i}(1+\abs{\xi})}\scal[(\nabla_{Je_i}S_{\Phi})]{\xi}{e_i}
+(-1)^{\abs{e_i}(1+\abs{\xi})}\scal[S_{\Phi}]{\nabla_{Je_i}\xi}{e_i}\\
&=\div S_{\Phi}[\xi]+(-1)^{\abs{e_i}\abs{\xi}}\scal[S_{\Phi}]{\nabla_{e_i}\xi}{Je_i}
\end{align*}
The second term can be transformed as follows. We use relabelling of summation indices
as well as exchange of $e_k$ and $Je_k$ with appropriate sign such that
\begin{align*}
(2)
&=(-1)^{\abs{e_i}\abs{\xi}}\scal[S_{\Phi}]{e_j\cdot\scal[h]{Je_j}{\nabla_{e_i}\xi}}{Je_i}\\
&=(-1)^{\abs{e_i}\abs{\xi}+\abs{e_j}+\abs{e_j}\abs{e_i}+\abs{e_j}\abs{\xi}}\scal[h]{Je_j}{\nabla_{e_i}\xi}\scal[S_{\Phi}]{e_j}{Je_i}\\
&=(-1)^{\abs{e_i}\abs{\xi}+\abs{e_j}\abs{e_i}+\abs{e_j}\abs{\xi}}\frac{1}{2}
\left((-1)^{\abs{e_j}}\scal[h]{Je_j}{\nabla_{e_i}\xi}\scal[S_{\Phi}]{e_j}{Je_i}\right.\\
&\qquad\qquad\qquad\qquad\qquad\qquad\left.+(-1)^{\abs{e_j}+\abs{e_i}\abs{e_j}}\scal[h]{e_i}{\nabla_{Je_j}\xi}\scal[S_{\Phi}]{e_j}{Je_i}\right)\\
&=\frac{1}{2}(-1)^{\abs{e_j}}\left((-1)^{\abs{e_i}\abs{\xi}}\scal[h]{\nabla_{e_i}\xi}{Je_j}
+(-1)^{\abs{e_j}\abs{\xi}+\abs{e_i}\abs{e_j}}\scal[h]{\nabla_{Je_j}\xi}{e_i}\right)\scal[S_{\Phi}]{e_j}{Je_i}\\
&=\frac{1}{2}(-1)^{\abs{e_j}}\scal[L_{\xi}h]{e_i}{Je_j}\scal[S_{\Phi}]{e_j}{Je_i}
\end{align*}
The last equation holds by Lem. \ref{lemLieDerivativeLeviCivita}.
\end{proof}

\begin{Thm}[Noether]
Let $\xi\in\mS M$ be a Killing vector field and $\Phi$ be superharmonic.
Then $\div Y_{\xi}=0$ vanishes.
\end{Thm}

\begin{proof}
By Lem. \ref{lemDivY} and the Killing property of $\xi$, we have $\div Y_{\xi}=\div S_{\Phi}[\xi]$.
The statement now follows directly from Lem. \ref{lemStressEnergy}.
\end{proof}

\bibliographystyle{alpha}

\end{document}